\documentclass[onecolumn]{IEEEtran}
\usepackage{graphicx}
\usepackage{amsmath,amsfonts,amssymb,amsthm}
\usepackage{fullpage}

\usepackage{array}
\usepackage{subfigure}
\usepackage{textcomp}
\usepackage{stfloats}
\usepackage{url}
\usepackage{verbatim}
\newtheorem{theorem}{Theorem}
\newtheorem{lemma}{Lemma}
\newtheorem{corollary}{Corollary}
\newtheorem{assumption}{Assumption}
\usepackage{algorithm}
\usepackage[noend]{algorithmic}
\usepackage{caption}
\newlength\myindent
\DeclareCaptionLabelSeparator{twospace}{: }
\def\BibTeX{{\rm B\kern-.05em{\sc i\kern-.025em b}\kern-.08em
    T\kern-.1667em\lower.7ex\hbox{E}\kern-.125emX}}
\usepackage{hyperref}
\hypersetup{hypertex=true,
            colorlinks=true,
            linkcolor=blue,
            anchorcolor=blue,
            citecolor=blue}

\usepackage{balance}
\newcommand{\E}[1]{\mathbb{E}\left[{#1}\right]}
\allowdisplaybreaks[1]
\begin{document}

\title{ABS: Adaptive Bounded Staleness Converges Faster and Communicates Less}
\author{Qiao Tan, Yun Gao}

\maketitle

\begin{abstract}
The convergence time and communication rounds are critical performance metrics in distributed learning with a parameter-server (PS) setting. Local stochastic gradient descent (SGD) is a widely used method for improving communication efficiency in distributed learning. However, synchronous local SGD can experience substantial slowdowns, forcing stragglers to perform multiple local updates and the other nodes to keep waiting. Asynchronous methods are immune to these slowdowns but can suffer from gradient staleness, leading to suboptimal results or even divergence. To address these challenges, we propose a novel asynchronous strategy named adaptive bounded staleness (ABS). ABS leverages two key enablers. Firstly, the number of workers that the PS waits for per round for gradient aggregation is adaptively selected to strike a balance between straggling and staleness. Secondly, workers with relatively high staleness are prompted to initiate a new round of computation, alleviating the negative effects of staleness. Simulation results demonstrate that ABS outperforms state-of-the-art schemes in terms of wall-clock time and communication rounds.
\end{abstract}


\section{Introduction}
As the size of datasets and the complexity of machine learning tasks continue to grow exponentially, stochastic gradient descent (SGD) methods have replaced traditional gradient descent (GD) methods and become the workhorse for large-scale machine learning \cite{bottou2010large}. To further accelerate the learning process, parallelism is leveraged through distributed implementation with multiple worker nodes, known as distributed learning \cite{r4}.

One commonly used setting in distributed SGD is the parameter-server (PS) architecture \cite{r1}, which consists of a central parameter server and distributed worker nodes. In this setting, workers download the latest global model from the PS, perform local computations, and upload the results to the PS to update the global model. This process is repeated until a convergence criterion is met. The primary schemes employed in the PS framework can be categorized into synchronous and asynchronous methods.

In synchronous SGD (SSGD) \cite{r2,r3,r9}, the server waits for all workers to complete their computations before updating the global model. Numerous studies have shown that SSGD can achieve high speedup and accuracy \cite{r1,r2,r3,r9}. However, the wall-clock time performance of SSGD is severely impacted by slow workers, also known as stragglers \cite{r4}. To address this issue, the $K$-sync SGD algorithm was proposed \cite{r5}, where the PS only waits for the $K$ fastest workers. Asynchronous methods, discussed in the following paragraph, were developed to further expedite the learning process.


Asynchronous SGD (ASGD) completely eliminates the influence of stragglers by allowing the PS to update the global model immediately after any worker finishes its local computation and uploads it to the server \cite{agarwal2011distributed}. It is evident that the time elapsed per training round in ASGD is significantly reduced compared to SSGD. However, the robustness to stragglers introduces the issue of staleness, where the gradient used to update the global parameter may not be consistent with the one used to compute it. This results in instability during the training process and an increased error floor compared to SSGD schemes \cite{r5,l4}. To mitigate these issues, Staleness-Aware (SA) methods were proposed \cite{SA}, which penalize the step size of stale gradients linearly based on their delay. This approach has been adopted by subsequent works \cite{SA2,SA3} and has become a common method for handling stale gradients. Gap-Aware (GA) methods were also proposed \cite{GA}, which penalize stale gradients linearly based on the gap between the worker model and server model. Additionally, the $K$-async SGD scheme was introduced \cite{r8} to reduce the impact of gradient staleness. $K$-async SGD is the asynchronous version of the $K$-sync SGD method, where the PS waits for the $K$ fastest workers to update the global model while the remaining workers continue computing and their stale gradients are used in subsequent rounds.

However, the issue of gradient staleness still poses a threat to algorithm convergence, as there is a risk of aggregating excessively stale gradients that can negatively impact performance. To tackle this challenge, we propose a novel asynchronous strategy called Adaptive Bounded Staleness (ABS). ABS aims to adaptively bound the staleness of gradients, hence its name. In ABS, similar to AdaSync, the number of workers that the server waits for during different stages of training adjusts adaptively based on the learning process. This approach strikes a balance between synchronous and asynchronous methods. Additionally, at the beginning of each training round, workers with a staleness exceeding a threshold value are required to restart their computation, effectively bounding the staleness. The threshold also adapts dynamically as the tolerance for staleness decreases throughout the training process. To further reduce communication overhead between workers and the PS, we incorporate the local SGD method in training. Unlike synchronizing models across all workers, local SGD allows each worker's model to evolve independently, with occasional model averaging. This approach has gained popularity in training deep neural networks \cite{l1, l2, l3, l4} and federated learning \cite{fedavg} due to its significant improvement in communication efficiency.

Furthermore, we provide a theoretical analysis of the convergence rate of the ABS algorithm under a non-convex condition. Simulation results demonstrate the advantages of ABS over state-of-the-art schemes in terms of both wall-clock time and communication rounds.

\textbf{Notation.} $\mathbb{R}$ denotes the real number fields; $\nabla f$ denotes the gradient operation for function $f$; $\bigcup$ denotes the union of sets. For sets $\mathcal{A}$ and $\mathcal{B}$,  $\mathcal{A}\subseteq\mathcal{B}$ represents that set $\mathcal{A}$ is a subset of set $\mathcal{B}$; $|\mathcal{A}|$ denotes the size of set $\mathcal{A}$.

\section{Problem Formulation}
\subsection{System Model}

We start by introducing the problem formulation under study. Our objective is to minimize the following empirical risk function: 
\begin{equation}\label{E2}
  F(\mathbf{w}) = \frac{1}{M} \sum_{m=1}^{M} f(\mathbf{w};\xi_{m}),
\end{equation}
where $\mathbf{w}\in\mathbb{R}^d$ is the $d$-dimensional model parameter to be optimized; and $f(\mathbf{w};\xi_{m})$ is a smooth loss function of parameter $\mathbf{w}$ with sample $\xi_m$ from dataset $\mathcal{D}=\{\xi_i\}_{i=1}^M$.

To implement the distributed local-SGD-based framework, we consider a PS-based architecture consisting of $N$ workers, denoted by the set $\mathcal{N}=\{1,..., N\}$. Each worker $n\in \mathcal{N}$ maintains a local dataset $\mathcal{D}_n$ drawn independently from the global dataset $\mathcal{D}$ without replacement, i.e., we have $\mathcal{D}=\bigcup_{n\in\mathcal{N}}\mathcal{D}_n$.

In order to accelerate the training process, an asynchronous aggregation mechanism of local SGD is investigated. More precisely, in the first round, each worker $n\in \mathcal{N}$ starts to perform $U$ local updates by setting the local model as $\mathbf{w}_n^{0,0}=\mathbf{w}^0$, where $\mathbf{w}^0$ is the initialized global model. Then, at each round $t\geq 0$, to update the global model $\mathbf{w}^{t+1}$, the PS aggregates the first $K^t$ local updates from the workers in the set $\mathcal{K}^t$. Note that we have $|\mathcal{K}^t| = K^t$. As a result, the local update rule of each worker $n$ can be given as
\begin{equation}
    \mathbf{w}_n^{t,u+1} = \mathbf{w}_n^{t,u} - \eta_t \frac{1}{B}\sum_{b=1}^B\nabla f(\mathbf{w}_n^{t,u}; \xi_{n,b}^{t, u}),\label{local update}
\end{equation}
for $u=0, ..., U-1$, where $U$ is the number of local iterations; $\eta_t$ is the stepsize; $B$ denotes the batch size; $\xi_{n,b}^{t, u}$ is drawn independently across all workers, local iterations and training rounds; and $\mathbf{w}_n^{t,u}$ is the local model of worker $n$ at local iteration $u$. Due to the fact that the gradients returned by each worker might be computed at a stale value of the global parameter $\mathbf{w}^{t}$, we use the variable $\tau_n^t$ to denote the staleness information at round $t$, i.e., work $n$ starts the computation with the local model given as $\mathbf{w}_n^{t-\tau_n^t,0}=\mathbf{w}^{t-\tau_n^t}$. Note that when $\tau_n^t=0$, it means that worker $n$ feeds back to update $\mathbf{w}^{t+1}$ at round $t$ after computing with the received global model $\mathbf{w}^{t}$. Moreover, each work $n\in \mathcal{K}^t$ sets the local model as $\mathbf{w}_n^{t+1,0}=\mathbf{w}^{t+1}$ to perform new local updates. As a result, the global model can be updated as 
\begin{equation}
     \mathbf{w}^{t+1}=\mathbf{w}^t-\frac{\eta_t}{K^t}\sum_{k \in \mathcal{K}^t } \sum_{u=0}^{U-1}\frac{1}{B}\sum_{b=1}^B\nabla f(\mathbf{w}_k^{t-\tau_k^t,u}; \xi_{k,b}^{t-\tau_k^t, u}).\label{update}
\end{equation}

\subsection{Background}
In this subsection, we extend the system model to incorporate state-of-the-art approaches as baseline algorithms. These baseline algorithms will serve as a point of comparison against our proposed method, which will be discussed in Section IV.

\textbf{Synchronous SGD.} Synchronous SGD (SSGD) is one of the most commonly adopted methods in distributed SGD \cite{r1}. In SSGD, the PS waits for all workers to finish their gradient computation before updating the global model, i.e., we have $K^t=N$, $U=1$, and $\tau_k^t=0$. Nevertheless, as stated before, though with stable performance and high accuracy, straggling becomes the main bottleneck for synchronous methods, which induces the development of asynchronous schemes.

\textbf{Asynchronous SGD.} 
In the primitive asynchronous SGD (ASGD) method \cite{agarwal2011distributed}, the PS updates the global model as soon as it receives the gradient from any worker, with $K^t=1$ and $U=1$. Although ASGD is robust to stragglers and converges faster than SSGD, it often suffers from a higher error floor due to gradient staleness. To enhance performance, one commonly used method is SA \cite{SA}, which penalizes the step size of stale gradients linearly based on their delay. Additionally, AdaSync \cite{AdaSync} proposes an adaptive approach that dynamically adjusts the number $K^t$ of workers the PS waits for to update the global parameters. However, the unbounded gradient staleness in asynchronous methods may lead to performance fluctuations. Hence, we introduce the ABS algorithm in the following section, which adaptively bounds the gradient staleness to mitigate such fluctuations.

\textbf{Local SGD.}
In both synchronous SGD and asynchronous SGD, if the workers are permitted to perform local updates before uploading their parameters, they employ the well-established local SGD method \cite{l2} \cite{l3} \cite{l4}, i.e., we have $U>1$. In comparison to $U=1$, local SGD allows the models to evolve locally on each worker node and only performs occasional model averaging.  This training approach has gained widespread adoption in recent years for training deep neural networks \cite{l1, l2, l3, l4} and in federated learning \cite{fedavg}, primarily due to its substantial improvement in communication efficiency. Taking into account both training speed and communication overhead, in ABS, we incorporate local training in each worker to demonstrate its performance.

\section{ABS: Adaptive Bounded Staleness}

In this section, we first introduce the asynchronous-based adaptive bounded staleness (ABS) method. The main idea of ABS is to adaptively bound the gradient staleness and adjust the number of workers that the PS waits for. By doing so, ABS achieves a delicate equilibrium between training speed and communication load. We then outline the key principles of ABS and discuss the analysis of algorithm hyperparameters that influence its performance.

\textit{Algorithm description:} 
To elaborate on the staleness of all the $N$ workers at round $t$, we define a $N$-length vector $\boldsymbol{\tau}_N^t=[\tau_1^t,\cdots, \tau_N^t]$. This vector is maintained by the parameter server (PS) and initialized with zero values. Each staleness element $\tau_n^t$ represents the number of consecutive rounds that worker $n$ has not received the global model from the PS. 
To identify the workers with high staleness, we define a threshold as $\tau_{max}^t$. At round $t$, if $\tau_n^t > \tau^t_{max}$, worker $n$ is forced to receive the latest global model from the PS and initiate a new round of computation. As the learning process advances, the system's tolerance for staleness decreases. Therefore, we dynamically adjust $\tau_{\text{max}}^t$ given as
\begin{equation}
    \tau^t_{max} = \max\Big\{1, \frac{N}{K^t} + a\Big\},
\end{equation}
where $a$ is some constant. In other words, as the number of workers to wait increases, the staleness threshold decreases gradually during the training process to control the staleness of the gradients. Additionally, we set the minimum value of $\tau^t_{\text{max}}$ to be 1 to ensure a minimum level of freshness. Moreover, following a similar approach as described in \cite{AdaSync}, we dynamically increase $K^t$ based on the learning process, given as
\begin{equation}
    \label{K}
    K^t = K^{0} \sqrt{\frac{f(\mathbf{w}^{0})}{f(\mathbf{w}^t)}},
\end{equation}
where $K^0$ is the number of workers we wait for in the initial round, $f(\mathbf{w}^{0})$ is the loss of the initial global model $\mathbf{w}^{0}$ and $f(\mathbf{w}^{t})$ is the loss of global model at iteration $t$. 
In general, at iteration $t$, we calculate the average loss of $K^t$ workers who update the PS and consider it as the loss of the global model. Therefore, the update of $\tau_n^t$ at each round $t$ follows the formula
\begin{equation}
    \!\!\!\!\tau_n^{t+1}=
    \begin{cases}
      0, ~~~~~~\text{if worker $n \in \mathcal{K}^t$ or $\tau_n^t > \tau^t_{max}$,\!\!} \\
      \tau_n^t+1, \text{otherwise}.
    \end{cases}\label{tau}
\end{equation}

As a start, the PS broadcasts the initial global model $\mathbf{w}^0$ to all the workers in the set $\mathcal{N}$ for computation. At each round $t$, the PS leverages the uploading information of the fastest $K^t$ workers to update the global model via (\ref{update}), where $K^t$ changes according to equation (\ref{K}). After the PS updates the global model  $\mathbf{w}^t$ and the staleness vector $\boldsymbol{\tau}_N^t$, the $K^t$ fastest workers and any worker with $\tau_n^t> \tau^t_{max}$ receive the latest model $\mathbf{w}^{t+1}$, and their corresponding staleness factor $\tau_n^{t}$ are set to 0. The procedure repeats until some convergence criterion is satisfied. The whole procedure is summarized in Algorithm \ref{a1}.

\begin{algorithm}[t!]
    \caption{ABS Algorithm}
    \label{a1}
    \begin{algorithmic}[1]
    \STATE Initializes $K^0$, $\tau^0 = \frac{N}{K^0} + a$, $\mathbf{w}^{0}$, $t=0$
    \FOR{$t = 0, 1,\dots, T$}
    \STATE \textbf{Server executes:}
    \begin{ALC@g}
    \STATE Waits for the uploads from the $K^t$ fastest workers
    \STATE Updates the global model $\mathbf{w}^t$ via (\ref{update})
    \STATE Update the number $K^t$ of workers to wait for
    via (\ref{K})
    \STATE Updates the ages of all $N$ workers via (\ref{tau})
    \STATE Sends $\mathbf{w}^t$ to the $K^t$ fastest workers and workers with $\tau_n^t>\tau^t_{max}$, sets their ages $\tau_n^t$ to 0
    \STATE Sets $t=t+1$
    \end{ALC@g}

\STATE \textbf{Worker executes:}
    \begin{ALC@g}
    \FOR{worker $n \in \mathcal{N}$}
    \STATE Receive the latest global model $\mathbf{w}^t$ from the PS and sets $\mathbf{w}_n^{t,0}=\mathbf{w}^t$
    \FOR{$u = 0, 1, \dots, U$}
    \STATE Updates its local model via (\ref{local update})
    \IF{receives new $\mathbf{w}^t$}
    \STATE Sets $\mathbf{w}_n^{t,0}=\mathbf{w}^t$, then restarts local computation
    \ENDIF
    \ENDFOR 
    \STATE Sends the accumulative gradient $\mathbf{w}_n^{t,0}-\mathbf{w}_n^{t,U}$ and average local loss to the PS
    \ENDFOR
    \end{ALC@g}
    \ENDFOR
    \end{algorithmic}
    \end{algorithm}

\textit{Example.} Consider $N=8$ workers training a CNN model with $U=10, B=32, \eta = 0.1$, and different values of the threshold $\tau_{max}$ of gradient staleness. As illustrated in Figure \ref{th_diff}, utilizing limited bounded gradient staleness leads to significantly improved performance, even with a small value of $K$, compared to the unbounded case corresponding to $\tau_{max}=\infty$. Furthermore, we observe that when $\tau_{\text{max}}=1$, the convergence rate is slightly slower. This can be attributed to the fact that discarding too many computations can impede the progress of convergence. To address this issue, ABS directly controls the gradient staleness. It leverages relatively larger gradient staleness to accelerate the training process during the early stages, while selectively receiving gradients with smaller staleness to ensure a lower error floor later on. By effectively managing the gradient staleness, ABS achieves a balance between rapid initial progress and attaining a desirable error floor.

\begin{figure}[!htbp]
    \centering
    \includegraphics[width=9.5cm]{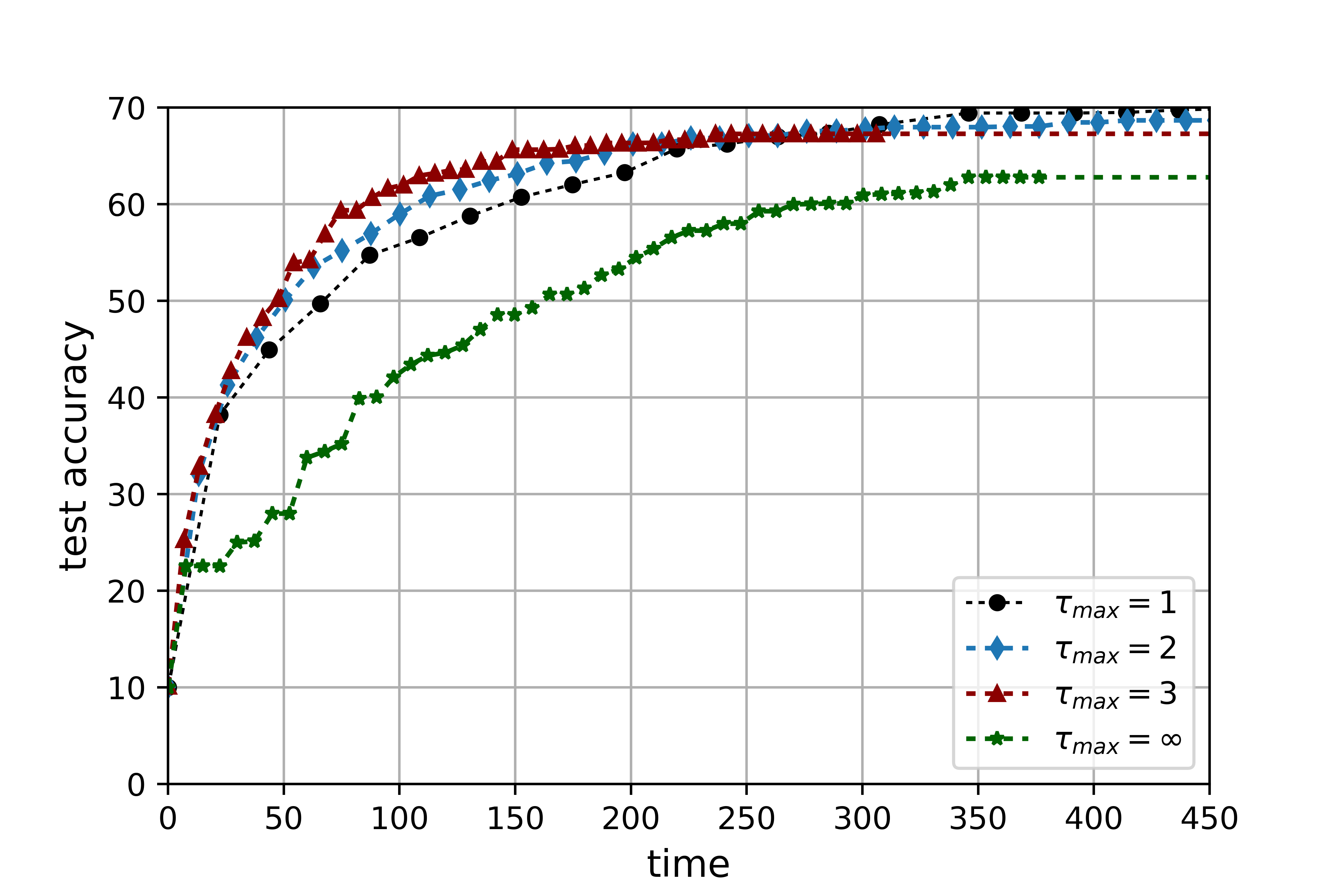}
     \caption{The performance of $K$-async ($K$=1) with $N=8$ and different values of the threshold $\tau_{max}$ of gradient staleness.} 
    \label{th_diff}
\end{figure}

\section{Convergence analysis}
In this section, we establish the convergence of the proposed ABS algorithm 
with a general (not necessarily convex) objective function. Our analysis is based on the following assumptions, which are widely adopted in related works such as \cite{AdaSync}.



\begin{assumption}\label{assum1} 
$F(\mathbf{w})$ is an L-smooth function, i.e., 
\begin{equation}
    \left \|\nabla F(\mathbf{w}_1) - \nabla F(\mathbf{w}_2) \right \|_2 \leq L\left \| \mathbf{w}_1 - \mathbf{w}_2\right \|_2 \quad \forall \mathbf{w}_1, \mathbf{w}_2.
\end{equation}
\end{assumption}

\begin{assumption} \label{assum2} 
The stochastic gradient is an unbiased estimate of the true gradient:
\begin{equation}
    \mathbb{E}_\xi[\nabla f(\mathbf{w}, \xi)] = \nabla F(\mathbf{w}).\label{assum2}
\end{equation}
\end{assumption}

\begin{assumption}\label{assum3} 
Let $g(\mathbf{w}_k^{t,u},\xi_k^{t,u})= \frac{1}{B}\sum_{b=1}^B\nabla f(\mathbf{w}_k^{t,u}, \xi_{k,b}^{t, u})$ be the batch gradient given $\mathbf{w}_k^{t,u}$ at iteration $t$ in local step $u$. We also assume that the variance of the stochastic batch gradient is also bounded as follows:
\begin{equation}
    \mathbb{E}_\xi[\left \| g(\mathbf{w}^{t,u}_k,\xi^{t,u}_k)-\nabla F(\mathbf{w}^{t,u}_k) \right \|_2^2] \leq \frac{\sigma^2}{m} + \frac{M_G}{m}\left \| \nabla F(\mathbf{w}^{t,u}_k) \right \|_2^2,
\end{equation}
where $\sigma^2$, $m > 0$ and $M_G > 0$ are some constants.
\end{assumption}



\begin{assumption}\label{assum4} 
With local training, we assume that for some  $\gamma < 1$, $\forall k$ and $\forall u$,
\begin{equation}
    \mathbb{E}\Big[\left \| \nabla F(\mathbf{w}^t)- \nabla F(\mathbf{w}_{k}^{t-\tau_k^t, u}) \right \|_2^2 \Big] \leq \gamma \mathbb{E}[\left \| \nabla F(\mathbf{w}^t)\right \|_2^2].
\end{equation}
\end{assumption}
Based on our assumptions, we have the following lemmas. 
\begin{lemma}
\label{lem1}
Let $g(\mathbf{w}^{t-\tau_k^t,u}_k,\xi^{t-\tau_k^t,u}_k)= \frac{1}{B}\sum_{b=1}^B\nabla f(\mathbf{w}_k^{t-\tau_k^t,u}; \xi_{k,b}^{t-\tau_k^t, u})$ be the batch gradient given $\mathbf{w}^{t-\tau_k^t,u}_k$ at iteration $t-\tau^t_k$ in local step u, and use the symbol $\mathbf{v}^{t,u}_k =  g(\mathbf{w}^{t-\tau_k^t,u}_k, \xi^{t-\tau_k^t,u}_k)$ to simply represent it. Then, we have
\begin{align}
    \E{\left \|\mathbf{v}^{t,u}_k - \nabla F(\mathbf{w}^t) \right \|^2_2} = \E{\left\|\mathbf{v}^{t,u}_k \right\|^2_2} - \E{\left\|\nabla F(\mathbf{w}^{t-\tau_k^t,u}_k)\right \|^2_2}+ \E{\left\|\nabla F(\mathbf{w}^t) - \nabla F(\mathbf{w}^{t-\tau_k^t,u}_k)\right\|_2^2}.
\end{align}
\begin{proof}
    The proof is presented in Appendix \ref{A1}.
\end{proof}
\end{lemma}

\begin{lemma}
\label{lem2}
Let $g(\mathbf{w}^{t-\tau_k^t,u}_k,\xi^{t-\tau_k^t,u}_k)= \frac{1}{B}\sum_{b=1}^B\nabla f(\mathbf{w}_k^{t-\tau_k^t,u}, \xi_{k,b}^{t-\tau_k^t, u})$. If the variance of the stochastic batch gradient is bounded as
\begin{equation}
    \mathbb{E}_\xi[\left \| g(\mathbf{w}^{t-\tau_k^t,u}_k,\xi^{t-\tau_k^t,u}_k)-\nabla F(\mathbf{w}^{t-\tau_k^t,u}_k) \right \|_2^2] \leq \frac{\sigma^2}{m} + \frac{M_G}{m}\left \| \nabla F(\mathbf{w}^{t-\tau_k^t,u}_k) \right \|_2^2, \nonumber
\end{equation}
Then the variance of the sum of the stochastic batch gradient is also bounded as follows:
\begin{equation}
    \mathbb{E}_{\xi}[||\sum_{k \in \mathcal{K}^t}g(\mathbf{w}^{t-\tau_k^t,u}_k,\xi^{t-\tau_k^t,u}_k) ||_2^2] \leq \frac{K^t \sigma^2}{m} + (\frac{M_G}{m} + 1)\sum_{k \in \mathcal{K}^t}||\nabla F(\mathbf{w}^{t-\tau_k^t,u}_k)||_2^2
\end{equation}
\begin{proof}
    The proof is presented in Appendix \ref{A1}.
\end{proof}
\end{lemma}

Based on our assumptions and lemmas, we have the following theorem and corollary for convergence. 
\begin{theorem} \label{theo1}For non-convex objective function $F(\cdot)$, where $F^*=\min_{\mathbf{w}} F(\mathbf{w})$, for a fixed learning rate $\eta$ we have the following ergodic convergence result for ABS:
\begin{equation}
    \frac{1}{T}\sum_{t=0}^{T-1}\mathbb{E}[\left \| 
    \nabla F(\mathbf{w}^t) \right \|_2^2] \leq \frac{2(F(\mathbf{w^0})-F^*)}{T\eta U(1-\gamma)} + \frac{L\eta \sigma^2}{K^0m(1-\gamma)}.
\end{equation}
\begin{proof}
    The proof is presented in Appendix \ref{A2}.
\end{proof}
\end{theorem}

\begin{corollary}
\label{coll1}
Following Theorem 1, with a learning rate $\eta = \sqrt{\frac{2(F(\mathbf{w}^0)-F^*)K^0 m}{TUL
\sigma^2}}$ satisfying $\eta \leq \frac{1}{L(\frac{M_G}{K^0m}+ \frac{1}{K^t})}$, the output of ABS has the following ergodic convergence rate:
\begin{equation}
    \frac{1}{T}\sum_{t=0}^{T-1}\mathbb{E}[\left \| \nabla F(\mathbf{w}^t) \right \|_2^2] \leq
    \frac{2}{1-\gamma}\sqrt{\frac{2(F(\mathbf{w}^0)-F^*) L \sigma^2}{T U K^0 m }}.\label{convergence}
\end{equation}
\begin{proof}
    The proof is presented in Appendix \ref{A3}.
\end{proof}
\end{corollary}

Corollary \ref{coll1} states that the convergence rate achieves $\mathcal{O}(1/\sqrt{T K U})$ under the given assumptions, which is consistent with the works \cite{l4,l5}.

\section{SIMULATION RESULTS}
In this section, we evaluate the effectiveness of the proposed ABS against state-of-the-art schemes including AdaSync \cite{AdaSync}, classical Local SGD \cite{l2}, and SA-AdaSync, i.e., AdaSync with SA\cite{SA} which penalizes gradients linear to their staleness. 

\subsection{Experimental Setting}

 For our experiments, we utilize the CIFAR-10 dataset, which consists of 50,000 training images and 10,000 validation images. The neural network architecture includes two convolutional layers and three fully connected layers. Each mini-batch size $B$ is set to 32, the learning rate $\eta$ is set to 0.1, and the number of local updates $U$ is set to 10.
We conduct experiments in two scenarios, one with $N=10$ and the other with $N=20$. The execution times of the workers follow a $\gamma$-distribution \cite{gamma}, which is a commonly used model for task execution times and captures the presence of stragglers.

\begin{figure}[t!]
\centering  
\subfigure[]{
\label{Fig.sub.31}
\includegraphics[width=8.4cm,height = 5.6cm]{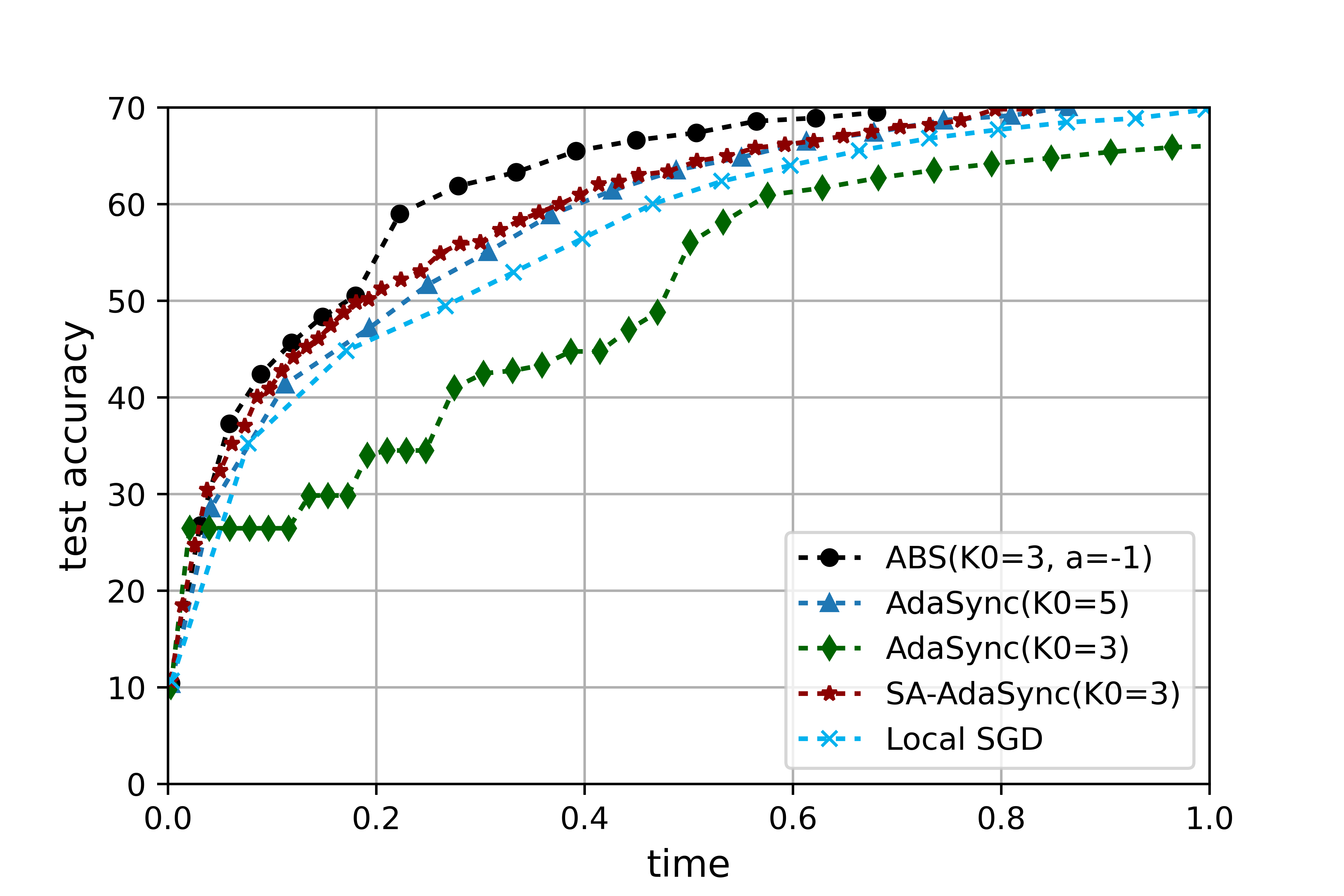}}\subfigure[]{
\label{Fig.sub.32}
\includegraphics[width=8.4cm,height = 5.6cm]{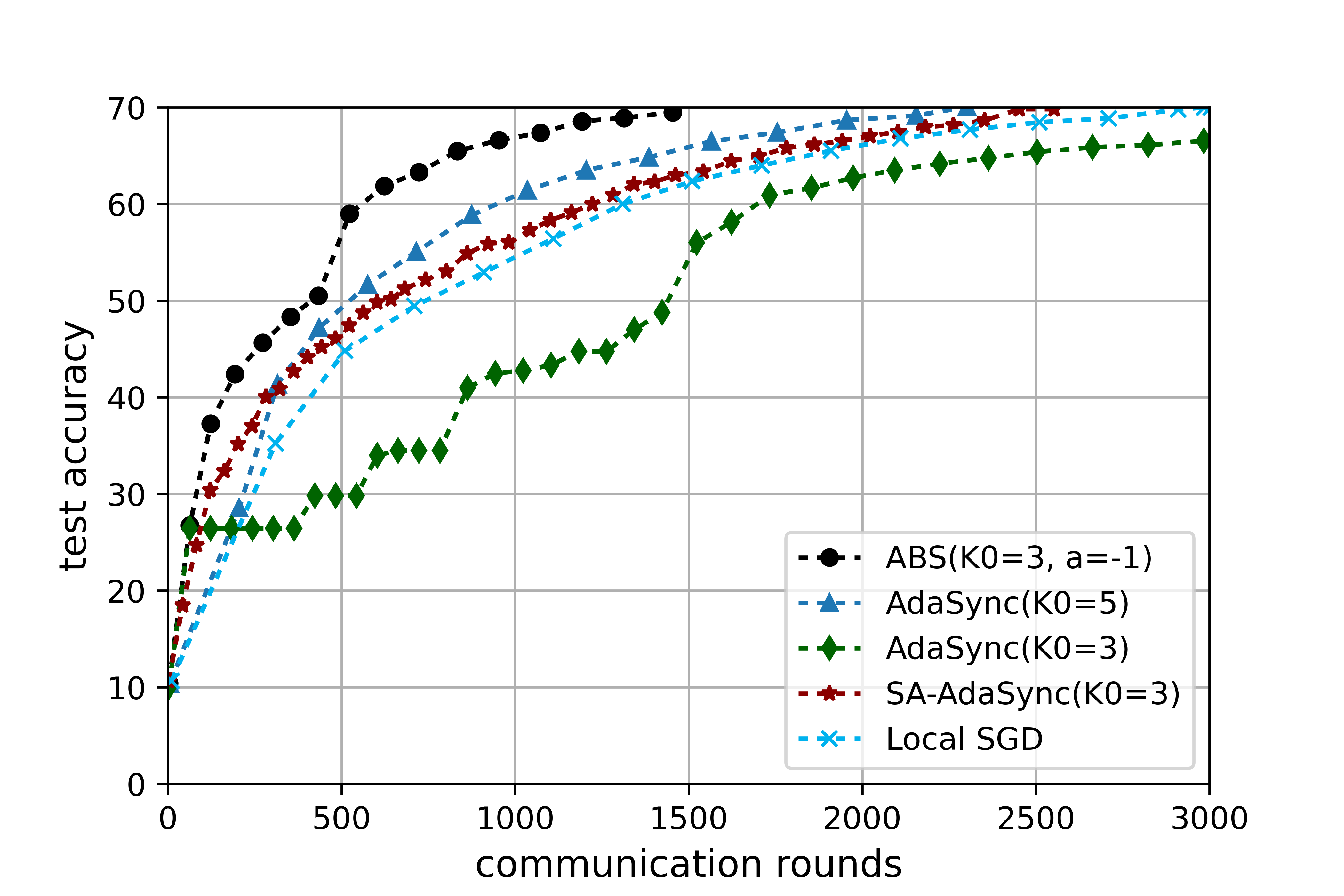}}
\caption{The performance of ABS compared to AdaSync, SA-AdaSync, and local SGD with $N=10$. (a) Learning accuracy vs. time. (b) Learning accuracy vs. communication rounds}
\label{fig1}
\end{figure}

\begin{figure}[t!]
\centering  
\subfigure[]{
\label{Fig.sub.41}
\includegraphics[width=8.4cm,height = 5.6cm]{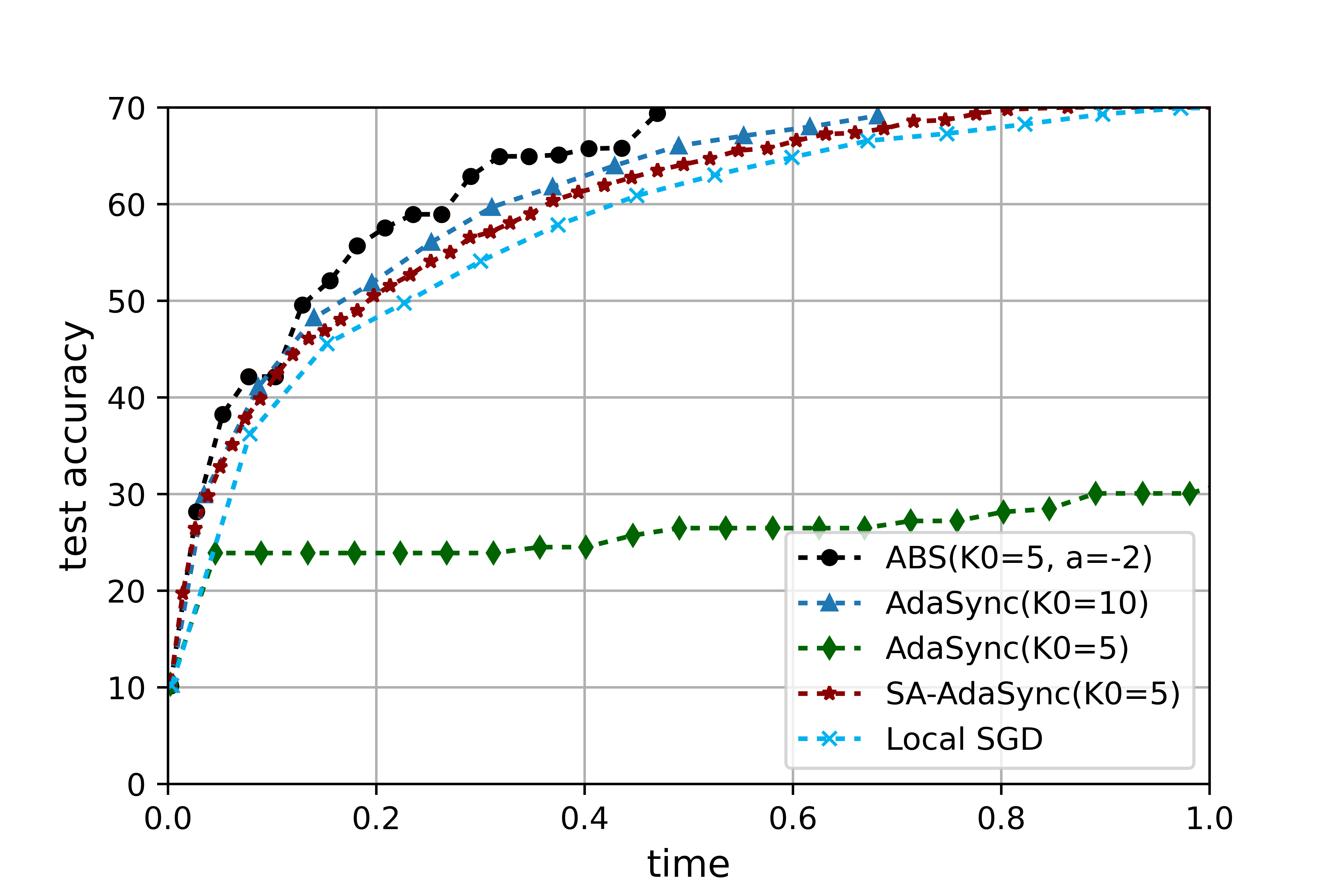}}\subfigure[]{
\label{Fig.sub.42}
\includegraphics[width=8.4cm,height = 5.6cm]{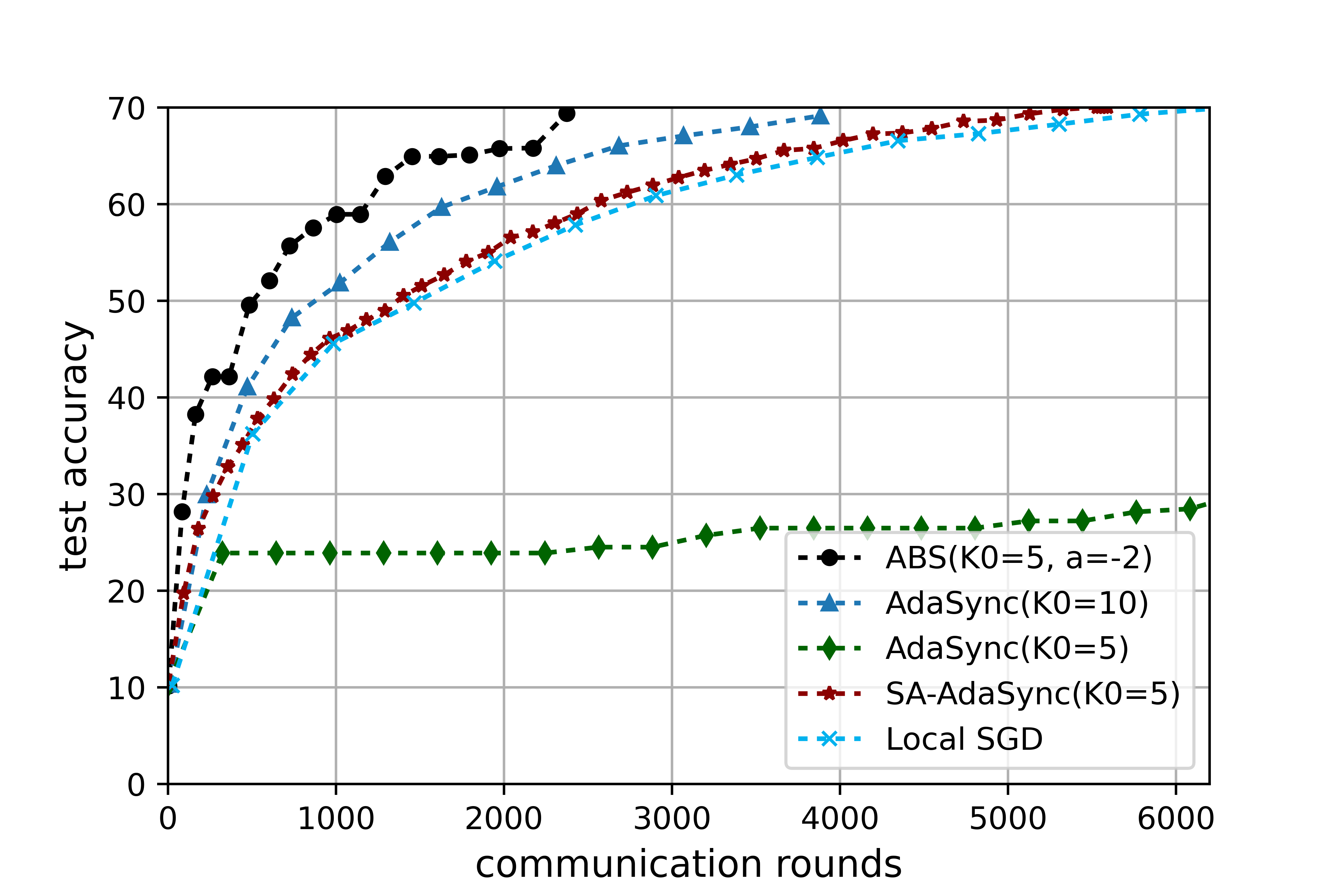}}
\caption{The performance of ABS compared to AdaSync, SA-AdaSync, and local SGD with $N=20$. (a) Learning accuracy vs. time. (b) Learning accuracy vs. communication rounds}
\label{fig2}
\end{figure}

\subsection{Performance Analysis}
\textbf{ABS Outperforms the State-of-Art Schemes.}
The performances of AdaSync, local SGD, SA-AdaSync, and ABS are evaluated in terms of the time and total communication rounds. It is worth noting that for each worker that communicates with the PS, the total communication rounds increase by 1. We consider two scenarios with $N=10$ and $N=20$, respectively.

Considering the scenario with $N = 10$, we evaluate the performance of AdaSync, local SGD, SA-AdaSync, and ABS in terms of both wall-clock time and communication rounds, as illustrated in Fig. \ref{fig1}(a) and (b), respectively. Remarkably, ABS demonstrates the fastest convergence rate while incurring the least communication cost, surpassing AdaSync even with different parameter values. This can be attributed to ABS's adoption of the bounded staleness approach, which allows for higher parallelism during the initial stages of training by starting with a smaller value of $K^0$. Additionally, by discarding some stale gradients, ABS effectively mitigates communication overload. This is further supported by ABS outperforming SA-AdaSync, which focuses on penalizing stale gradients. Thus, while there may be a loss of some gradients through discarding, the results indicate that excessively stale gradients are not essential for achieving desirable outcomes. This observation holds true even as the number of workers increases, as demonstrated in Fig. \ref{fig2}. This further supports the superiority of ABS in terms of convergence rate and communication efficiency compared to other methods.

For the scenario with $N = 10$, the impact of different parameters $a$ in ABS is illustrated in Figure \ref{fig3}. When $a$ is large, ABS gradually approaches the behavior of AdaSync, as seen with $a=0$. In this case, $\tau_{\text{max}}^t$ becomes large enough that very few stale gradients are discarded. Consequently, due to the significant gradient staleness, ABS performs similarly to AdaSync. On the other hand, when $a$ is small, such as $a=-3$, $\tau_{\text{max}}^t$ decreases to a value that is considered relatively insignificant. In this scenario, ABS discards a large number of gradients, which results in a slower convergence process. Taking both time and communication overhead into account, we conclude that $a=-2$ is the optimal choice for this particular example.

These findings emphasize the importance of appropriately selecting the parameter $a$ in ABS, as it directly impacts the staleness threshold $\tau_{\text{max}}^t$ and ultimately affects the trade-off between convergence speed and communication efficiency.

\begin{figure}[t!]
\centering  
\subfigure[]{
\label{Fig.sub.1}
\includegraphics[width=8.4cm,height = 5.6cm]{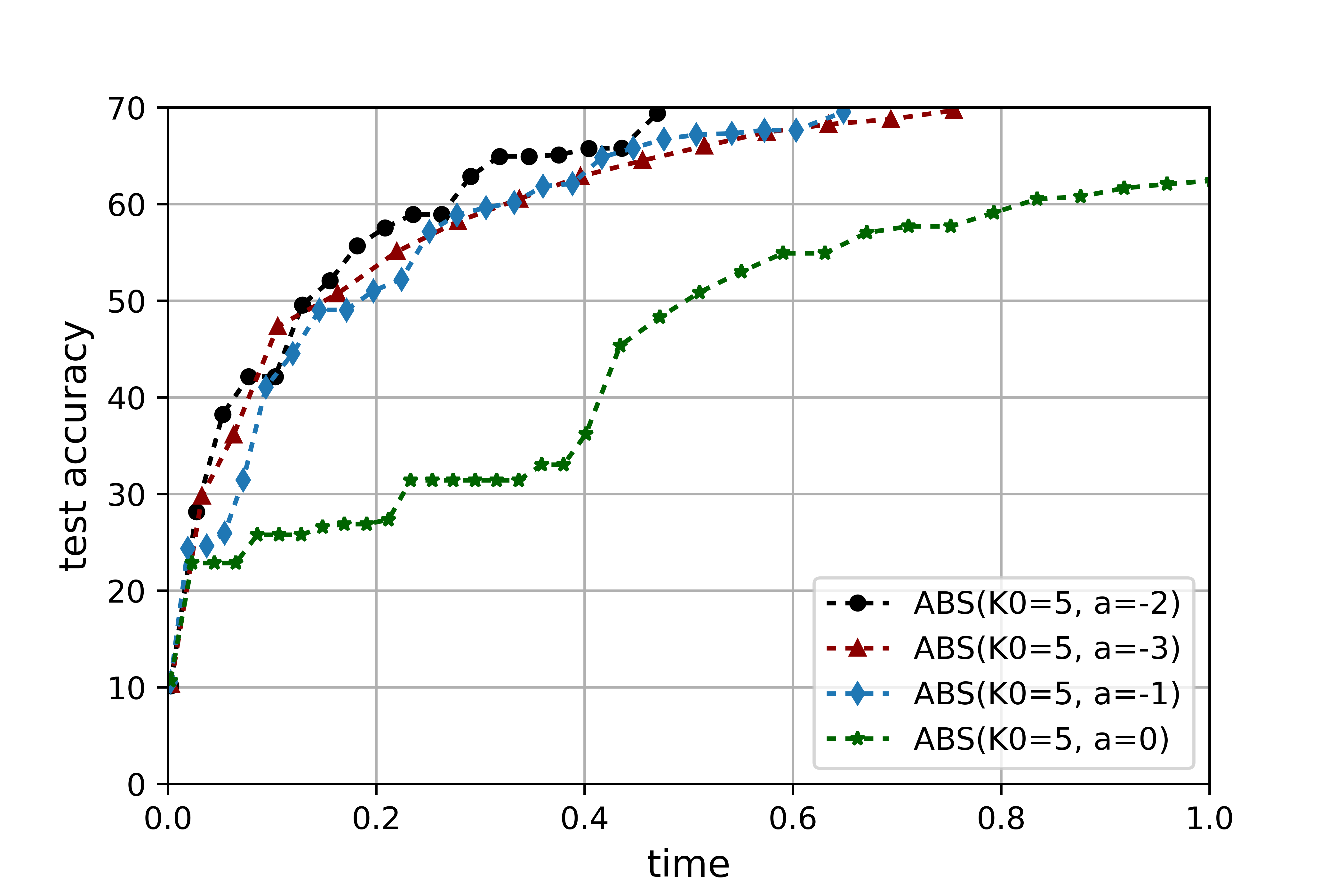}}\subfigure[]{
\label{Fig.sub.2}
\includegraphics[width=8.4cm,height = 5.6cm]{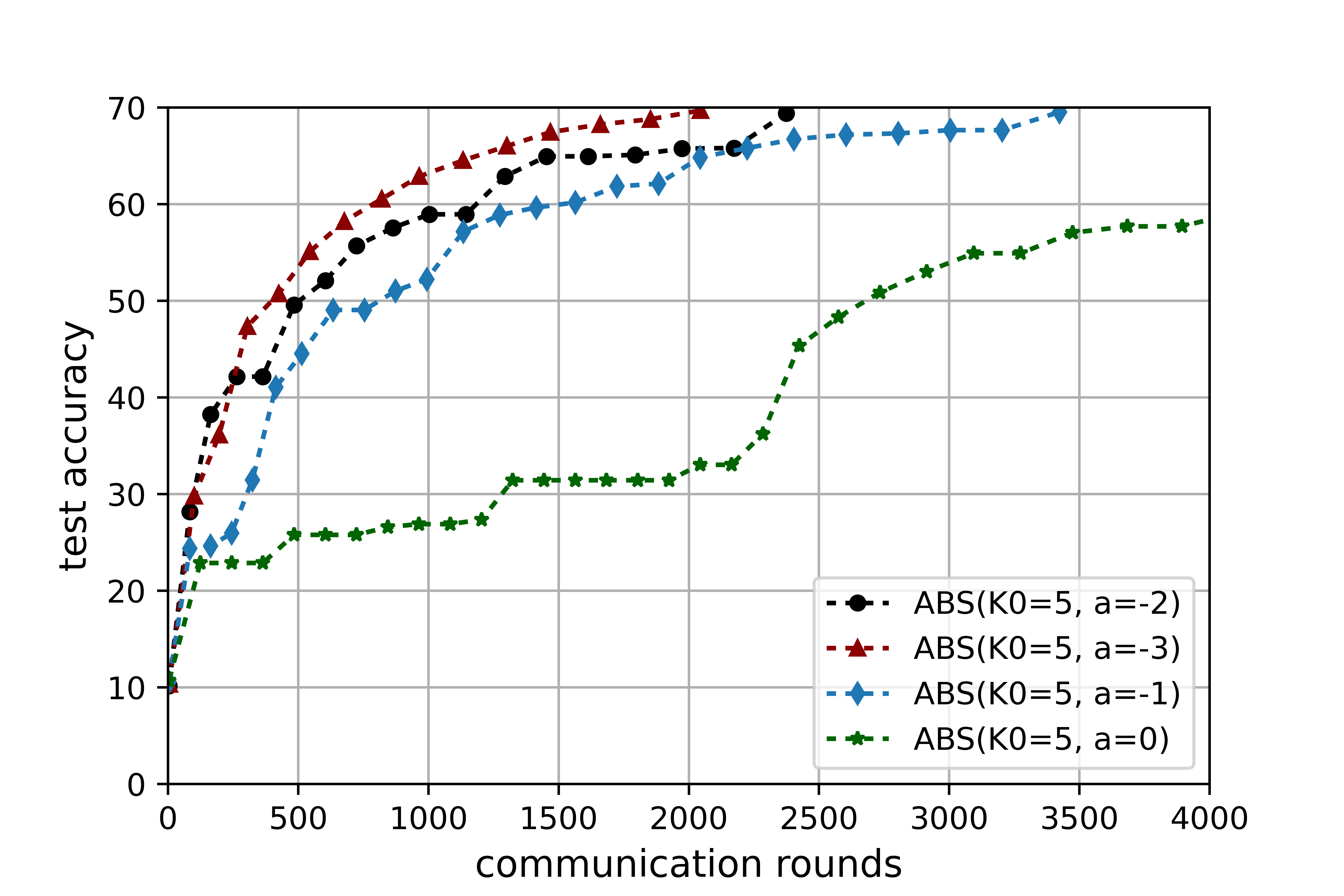}}
\caption{The performance of ABS with different parameter $a$ when $K^0 = 5$ and $N = 20$. (a) Learning accuracy vs. time. (b) Learning accuracy vs. communication rounds}
\label{fig3}
\end{figure}

\section{conclusion}

This paper has presented a novel approach, Adaptive Bounded Staleness (ABS), to enhance the training time and communication efficiency of distributed SGD while maintaining high accuracy. By leveraging adaptive thresholding and controlling gradient staleness, ABS achieves superior performance compared to the existing methods. The key idea behind ABS is to utilize relatively large gradient staleness during the early training stages to accelerate the process while receiving only gradients with small staleness later on to ensure higher accuracy. Through extensive experiments, we have confirmed the effectiveness of ABS in improving convergence rates and communication efficiency. 

Furthermore, our study explores the possibility of discarding gradients with large staleness, revealing that their impact on performance is minimal. This finding supports the effectiveness of ABS in managing gradient staleness and optimizing the training process. Overall, ABS provides a promising solution for distributed SGD, offering improved training time, communication efficiency, and high accuracy. Further research can be conducted to explore its application in other domains and to investigate additional parameter configurations for different scenarios.

\bibliographystyle{IEEEtran}
\bibliography{mybib}
\appendix
\section{Proof}
\label{appendix1}
\subsection{Proof of lemmas}
\label{A1}

\noindent Proof of \textbf{Lemma} \ref{lem1}
\begin{proof}
\begin{align}
\label{el1}
    \E{\left \|\mathbf{v}^{t,u}_k - \nabla F(\mathbf{w}^t) \right \|^2_2} &=  \E{\left \|\mathbf{v}^{t,u}_k -\nabla F(\mathbf{w}^{t-\tau_k^t,u}_k) + \nabla F(\mathbf{w}^{t-\tau_k^t,u}_k) - \nabla F(\mathbf{w}^t) \right \|^2_2} \nonumber\\
    &=\E{\left \|\mathbf{v}^{t,u}_k -\nabla F(\mathbf{w}^{t-\tau_k^t,u}_k)\right \|^2_2} +\E{\left\| \nabla F(\mathbf{w}^{t-\tau_k^t,u}_k) - \nabla F(\mathbf{w}^t) \right \|^2_2}.
\end{align}
The last line holds since the cross term is 0 as derived below.
\begin{align*}
    &\E{ (\mathbf{v}^{t,u}_k -\nabla F(\mathbf{w}^{t-\tau_k^t,u}_k))^T ( \nabla F(\mathbf{w}^{t-\tau_k^t,u}_k) - \nabla F(\mathbf{w}^t))} \nonumber\\
    &\overset{(a1)}{=} \E{\mathbb{E}_\xi\left[(\mathbf{v}^{t,u}_k-\nabla F(\mathbf{w}^{t-\tau_k^t,u}_k))^T\right]} \E{(\nabla F(\mathbf{w}^{t-\tau_k^t,u}_k) - \nabla F(\mathbf{w}^t))} \nonumber \\
    &\overset{(b1)}{=}\E{(\mathbb{E}_\xi[\mathbf{v}^{t,u}_k] - \nabla F(\mathbf{w}^{t-\tau_k^t,u}_k))^T}  \E{\nabla F(\mathbf{w}^{t-\tau_k^t,u}_k) - \nabla F(\mathbf{w}^t)}\nonumber\\
    &\overset{(c1)}{=}0
\end{align*}
Here ($a1$) and ($b1$) indicate that we first calculate the expectation for $\xi$. Step ($c1$) follows from that $\mathbb{E}_{\xi}[\mathbf{v}^{t,u}_k] =  \nabla F(\mathbf{w}^{t-\tau_k^t,u}_k)$.

Returning to (\ref{el1}), the first term can be written as, 
\begin{align}
\label{el2}
\E{\left \|\mathbf{v}^{t,u}_k -\nabla F(\mathbf{w}^{t-\tau_k^t,u}_k)\right \|^2_2}  &\overset{(a2)}{=} \E{\left \|\mathbf{v}^{t,u}_k\right \|^2_2} - 2\E{(\mathbf{v}^{t,u}_k)^T \nabla F(\mathbf{w}^{t-\tau_k^t,u}_k)} + \E{\left \| \nabla F(\mathbf{w}^{t-\tau_k^t,u}_k)\right \|_2^2} \nonumber \\
&\overset{(b2)}{=}\E{\left \|\mathbf{v}^{t,u}_k\right \|^2_2} - 2\E{\mathbb{E}_\xi[(\mathbf{v}^{t,u}_k)^T \nabla F(\mathbf{w}^{t-\tau_k^t,u}_k)]}  + \E{\left \| \nabla F(\mathbf{w}^{t-\tau_k^t,u}_k)\right \|_2^2} \nonumber \\
&\overset{(c2)}{=} \E{\left \|\mathbf{v}^{t,u}_k\right \|^2_2} - 2\E{\left \| \nabla F(\mathbf{w}^{t-\tau_k^t,u}_k)\right \|_2^2} + \E{\left \| \nabla F(\mathbf{w}^{t-\tau_k^t,u}_k)\right \|_2^2}\nonumber\\
&=\E{\left \|\mathbf{v}^{t,u}_k\right \|^2_2} - \E{\left \| \nabla F(\mathbf{w}^{t-\tau_k^t,u}_k)\right \|_2^2}. 
\end{align}
Here ($a2$) follows from the equation $||a - b||_2^2= ||a||_2^2 - 2a^Tb + ||b||_2^2$ and ($b2$) means that we first calculate the expectation for $\xi$. Step ($c2$) is due to $\mathbb{E}_\xi[\mathbf{v}^{t,u}_k] = \nabla F(\mathbf{w}^{t-\tau_k^t,u}_k)$. Bring (\ref{el2}) to (\ref{el1}), we get Lemma \ref{lem1}.  
\end{proof}

\noindent Proof of \textbf{Lemma} \ref{lem2}

\begin{proof}
First let us consider the expectation of any cross term  such that $k \neq k' $, we can get that

\begin{align}
    &\mathbb{E}_{\xi}\left[(g(\mathbf{w}^{t-\tau_k^t,u}_k,\xi^{t-\tau_k^t,u}_k)-\nabla F(\mathbf{w}^{t-\tau_k^t,u}_k))^T (g(\mathbf{w}^{t-\tau_{k'}^t,u}_{k'},\xi^{t-\tau_{k'}^t,u}_{k'})-\nabla F(\mathbf{w}^{t-\tau_{k'}^t,u}_{k'}))\right] \nonumber\\
    &\overset{(a3)}{=}\mathbb{E}_{\xi^{t-\tau_{k'}^t,u}_{k'}|\xi^{t-\tau_{k}^t,u}_{k}}\left[\mathbb{E}_{\xi^{t-\tau_k^t,u}_k}\left[(g(\mathbf{w}^{t-\tau_k^t,u}_k,\xi^{t-\tau_k^t,u}_k)-\nabla F(\mathbf{w}^{t-\tau_k^t,u}_k))^T (g(\mathbf{w}^{t-\tau_{k'}^t,u}_{k'},\xi^{t-\tau_{k'}^t,u}_{k'})-\nabla F(\mathbf{w}^{t-\tau_{k'}^t,u}_{k'}))\right]\right] \nonumber\\
    &\overset{(b3)}{=} \mathbb{E}_{\xi^{t-\tau_{k'}^t,u}_{k'}} \left[(\nabla F(\mathbf{w}^{t-\tau_k^t,u}_k) - \nabla F(\mathbf{w}^{t-\tau_k^t,u}_k))^T(g(\mathbf{w}^{t-\tau_{k'}^t,u}_{k'},\xi^{t-\tau_{k'}^t,u}_{k'})-\nabla F(\mathbf{w}^{t-\tau_{k'}^t,u}_{k'}))\right]\nonumber\\
    &= 0.
\end{align}
Here step ($a3$) means we first calculate the expectation for $\xi^{t-\tau_k^t,u}_k$, then the expectation for $\xi^{t-\tau_{k'}^t,u}_{k'}$, and ($b3$) is due to $\mathbb{E}_{\xi^{t-\tau_k^t,u}_k}\left[g(\mathbf{w}^{t-\tau_k^t,u}_k,\xi^{t-\tau_k^t,u}_k)\right] = \nabla F(\mathbf{w}^{t-\tau_k^t,u}_k)$.
Thus the cross terms are all 0. So we can get that
\begin{align}\label{eq18}
    \mathbb{E}_{\xi}\left[||\sum_{k \in \mathcal{K}^t}(g(\mathbf{w}^{t-\tau_k^t,u}_k,\xi^{t-\tau_k^t,u}_k)-\nabla F(\mathbf{w}^{t-\tau_k^t,u}_k)) ||_2^2\right]&=\sum_{k \in \mathcal{K}^t}\mathbb{E}_{\xi}\left[||(g(\mathbf{w}^{t-\tau_k^t,u}_k,\xi^{t-\tau_k^t,u}_k)-\nabla F(\mathbf{w}^{t-\tau_k^t,u}_k)) ||_2^2\right]\nonumber\\
    &\overset{(a4)}{\leq} \sum_{k \in \mathcal{K}^t}(\frac{\sigma^2}{m} + \frac{M_G}{m}|| \nabla F(\mathbf{w}^{t-\tau_k^t,u}_k) ||_2^2),
\end{align}
where ($a4$) is from assumption \ref{assum3}. Thus, we have 
\begin{align}
     &\mathbb{E}_{\xi}\left[||\sum_{k \in \mathcal{K}^t}g(\mathbf{w}^{t-\tau_k^t,u}_k,\xi^{t-\tau_k^t,u}_k) ||_2^2\right]= \mathbb{E}_{\xi}\left[||\sum_{k \in \mathcal{K}^t}(g(\mathbf{w}^{t-\tau_k^t,u}_k,\xi^{t-\tau_k^t,u}_k) -  \nabla F(\mathbf{w}^{t-\tau_k^t,u}_k) + \nabla F(\mathbf{w}^{t-\tau_k^t,u}_k))||^2_2\right]\nonumber\\
     &\overset{(a5)}{=}  \mathbb{E}_{\xi}\left[||\sum_{k \in \mathcal{K}^t}(g(\mathbf{w}^{t-\tau_k^t,u}_k,\xi^{t-\tau_k^t,u}_k) -  \nabla F(\mathbf{w}^{t-\tau_k^t,u}_k)||^2_2\right] + \mathbb{E}\left[||\sum_{k \in \mathcal{K}^t}\nabla F(\mathbf{w}^{t-\tau_k^t,u}_k)||_2^2\right] \nonumber\\
     &\overset{(b5)}{\leq} \sum_{k \in \mathcal{K}^t}(\frac{\sigma^2}{m} + \frac{M_G}{m}|| \nabla F(\mathbf{w}^{t-\tau_k^t,u}_k) ||_2^2) + \sum_{k \in \mathcal{K}^t}||\nabla F(\mathbf{w}^{t-\tau_k^t,u}_k)||_2^2\nonumber\\
     &= \frac{K^t \sigma^2}{m} + (\frac{M_G}{m} + 1)\sum_{k \in \mathcal{K}^t}||\nabla F(\mathbf{w}^{t-\tau_k^t,u}_k)||_2^2
\end{align}
Here step ($a5$) follows from that $\mathbb{E}_{\xi}[(g(\mathbf{w}^{t-\tau_k^t,u}_k,\xi^{t-\tau_k^t,u}_k) -  \nabla F(\mathbf{w}^{t-\tau_k^t,u}_k))^T\nabla F(\mathbf{w}^{t-\tau_k^t,u}_k)] = 0$,  ($b5$) is from eq. (\ref{eq18}). So we complete the proof. 
\end{proof}

\subsection{Proof of \textbf{Theorem} \ref{theo1}}
\label{A2}
\begin{proof}
Let $g(\mathbf{w}^{t-\tau_k^t,u}_k,\xi^{t-\tau_k^t,u}_k)= \frac{1}{B}\sum_{b=1}^B\nabla f(\mathbf{w}_k^{t-\tau_k^t,u}, \xi_{k,b}^{t-\tau_k^t, u})$, $\mathbf{v}^{t,u} = \frac{1}{K^t}\sum_{k \in \mathcal{K}^t} g(\mathbf{w}^{t-\tau_k^t,u}_k, \xi^{t-\tau_k^t,u}_k)$. From Lipschitz continuity, we have the following.
\begin{align}
    F(\mathbf{w}^{t+1}) &\leq F(\mathbf{w}^{t}) +     (\mathbf{w}^{t+1} - \mathbf{w}^{t})^{T} \nabla F(\mathbf{w}^{t}) + \frac{L}{2} \left \|\mathbf{w}^{t+1} - \mathbf{w}^{t} \right \|_{2}^{2}  \nonumber\\
    &= F(\mathbf{w}^{t})- \frac{\eta}{K^t} \sum_{k\in \mathcal{K}^t} \sum_{u=0}^{U-1}g(\mathbf{w}^{t-\tau_k^t,u}_k, \xi^{t-\tau_k^t,u}_k)^T \nabla F(\mathbf{w}^t) + \frac{L}{2} \left \|\eta \sum_{u=0}^{U-1}\mathbf{v}^{t,u} \right \|_{2}^{2}  \nonumber \\ 
    &\overset{(a6)}{=} F(\mathbf{w}^{t}) - \frac{\eta}{2K^t} \sum_{k\in \mathcal{K}^t} \sum_{u=0}^{U-1} \left \|\nabla F(\mathbf{w}^t)\right\|^2_2 -\frac{\eta}{2K^t} \sum_{k \in \mathcal{K}^t}^{K^t} \sum_{u=0}^{U-1} \left \| g(\mathbf{w}^{t-\tau_k^t,u}_k, \xi^{t-\tau_k^t,u}_k) \right \|^2_2  \nonumber\\
    & \qquad  + \frac{\eta}{2K^t} \sum_{k\in \mathcal{K}^t} \sum_{u=0}^{U-1}
    \left \|g(\mathbf{w}^{t-\tau_k^t,u}_k, \xi^{t-\tau_k^t,u}_k) - \nabla F(\mathbf{w}^t) \right \|^2_2 + \frac{L\eta^2}{2} \left \| \sum_{u=0}^{U-1}\mathbf{v}^{t,u} \right \|_{2}^{2}.
\end{align}
Here ($a6$) follows from $2{a}^T{b} = ||{a}||_2^2 + ||{b}||_2^2 - ||{a}-{b}||_2^2$. We remove the restriction $ \tau_k^t \leq \tau^t_{max}$, for we prove that for all cases of $\tau_k^t$, it should also be satisfied under the restriction. Taking expectation,
\begin{align}
    \E{F(\mathbf{w}^{t+1})} &\leq \E{F(\mathbf{w}^t)} - \frac{\eta U}{2}\E{\left \| \nabla F(\mathbf{w}^t)\right\|^2_2} -  \frac{\eta}{2K^t}\sum_{k\in \mathcal{K}^t}\sum_{u=0}^{U-1}\E{ \left \|g(\mathbf{w}^{t-\tau_k^t,u}_k, \xi^{t-\tau_k^t,u}_k)\right\|^2_2}  \nonumber\\
    &\qquad + \frac{\eta}{2K^t} \sum_{k \in \mathcal{K}^t}\sum_{u=0}^{U-1}
    \E{\left \|g(\mathbf{w}^{t-\tau_k^t,u}_k, \xi^{t-\tau_k^t,u}_k)-\nabla F(\mathbf{w}^t) \right \|^2_2} + \frac{L\eta^2}{2} \E{\left \| \sum_{u=0}^{U-1}\mathbf{v}^{t,u} \right  \|_{2}^{2}}  \nonumber\\
    &\overset{(a7)}{=} \E{F(\mathbf{w}^t)} - \frac{\eta U}{2}\E{\left \|  \nabla F(\mathbf{w}^t)\right\|^2_2} - \frac{\eta}{2K^t}\sum_{k\in \mathcal{K}^t}\sum_{u=0}^{U-1}\E{ \left \|g(\mathbf{w}^{t-\tau_k^t,u}_k, \xi^{t-\tau_k^t,u}_k)\right\|^2_2}  \nonumber\\
    &\qquad+ \frac{\eta}{2K^t}\sum_{k\in \mathcal{K}^t}\sum_{u=0}^{U-1}\E{ \left \|g(\mathbf{w}^{t-\tau_k^t,u}_k, \xi^{t-\tau_k^t,u}_k)\right\|^2_2} - 
    \frac{\eta}{2K^t}\sum_{k \in \mathcal{K}^t}\sum_{u=0}^{U-1}\E{ \left \|\nabla F(\mathbf{w}^{t-\tau_k^t,u}_k)\right\|^2_2}  \nonumber\\
    &\qquad+\frac{\eta}{2K^t}\sum_{k\in \mathcal{K}^t}\sum_{u=0}^{U-1}\E{ \left \|\nabla F(\mathbf{w}^t) - \nabla F(\mathbf{w}^{t-\tau_k^t,u}_k)\right\|^2_2} + \frac{L\eta^2}{2} \E{ \left \| \sum_{u=0}^{U-1}\mathbf{v}^{t,u} \right \|_{2}^{2}} \nonumber\\
    &\overset{(b7)}{\leq}\E{F(\mathbf{w}^t)} - \frac{\eta U(1-\gamma)}{2}\E{\left \|\nabla
    F(\mathbf{w}^t)\right\|^2_2} - \frac{\eta}{2K^t}\sum_{k\in \mathcal{K}^t}\sum_{u=0}^{U-1}\E{ \left \|\nabla F(\mathbf{w}^{t-\tau_k^t,u}_k)\right\|^2_2}  \nonumber\\
    &\qquad\qquad\qquad\qquad\qquad\qquad\qquad\qquad\qquad\qquad\qquad + \frac{L\eta^2}{2} \E{ \left \| \sum_{u=0}^{U-1}\mathbf{v}^{t,u} \right \|_{2}^{2}}  \nonumber\\
    &\overset{(c7)}{\leq} \E{F(\mathbf{w}^t)} - \frac{\eta U(1-\gamma)}{2}\E{\left \|\nabla F(\mathbf{w}^t)\right\|^2_2} - \frac{\eta}{2K^t}\sum_{k\in \mathcal{K}^t}\sum_{u=0}^{U-1}\E{ \left \|\nabla F(\mathbf{w}^{t-\tau_k^t,u}_k)\right\|^2_2}  \nonumber\\
    &\qquad\qquad\qquad\qquad\qquad\qquad\qquad\qquad\qquad\qquad\qquad+ \frac{L\eta^2}{2}\sum_{u=0}^{U-1} \E{ \left \| \mathbf{v}^{t,u} \right \|_{2}^{2}}  \nonumber\\
    &\overset{(d7)}{\leq} \E{F(\mathbf{w}^t)} - \frac{\eta U(1-\gamma)}{2}\E{\left \|\nabla F(\mathbf{w}^t)\right\|^2_2} - \frac{\eta}{2K^t}\sum_{k\in \mathcal{K}^t}\sum_{u=0}^{U-1}\E{ \left \|\nabla F(\mathbf{w}^{t-\tau_k^t,u}_k)\right\|^2_2}  \nonumber\\
    &\qquad\qquad\qquad\qquad + \frac{L\eta^2}{2(K^{t})^2}\sum_{u=0}^{U-1}(\frac{K^t\sigma^2}{m} + (\frac{M_G}{m} + 1)\sum_{k\in \mathcal{K}^t}\E{\left \|\nabla F(\mathbf{w}^{t-\tau_k^t,u}_k) \right \|_{2}^2})  \nonumber\\
    &\leq  \E{F(\mathbf{w}^t)} - \frac{\eta U(1-\gamma)}{2}\E{\left \|\nabla F(\mathbf{w}^t)\right\|^2_2} + \frac{L\eta^2U\sigma^2}{2K^t m}  \nonumber\\
    & \qquad\qquad\qquad\qquad - \frac{\eta}{2K
    ^t}\sum_{u=0}^{U-1}\sum_{k\in \mathcal{K}^t}(1-L\eta(\frac{M_G}{K^t m} + \frac{1}{K^t}))\E{\left \| \nabla F(\mathbf{w}^{t-\tau_k^t,u}_k) \right\|^2_2}  \nonumber\\
    &\overset{(e7)}{\leq} \E{F(\mathbf{w}^t)} - \frac{\eta U(1-\gamma)}{2}\E{\left \|\nabla F(\mathbf{w}^t)\right\|^2_2} + \frac{L\eta^2U\sigma^2}{2K^t m}   
\end{align}

Here step ($a7$) follows from Lemma \ref{lem1} and step ($b7$) follows from assumption \ref{assum4} that
\begin{equation}
    \mathbb{E}\left[\left \| \nabla F(\mathbf{w}^t)- \nabla F(\mathbf{w}_{k}^{t-\tau_k^t, u}) \right \|_2^2\right] \leq \gamma \mathbb{E}\left[\left \| \nabla F(\mathbf{w}^t)\right \|_2^2\right]\nonumber
\end{equation}

for some constant $\gamma < 1$. Step ($c7$) follows from the equation $\E{||\sum_i x_i||_2^2} \leq \sum_i \E{||x_i||_2^2} $ . Step ($d7$) follows from Lemma \ref{lem2} and step ($e7$) follows from choosing $\eta < \frac{1}{L(\frac{M_G}{m K^t} + \frac{1}{K^t})}$.

Then after re-arrangement, we obtain the following: 
\begin{align}
   \E{\left \|\nabla F(\mathbf{w}^t)\right\|^2_2} \leq \frac{2(\E{F(\mathbf{w}^{t})} - \E{F(\mathbf{w}^{t+1})})}{\eta U (1-\gamma)} + \frac{L\eta\sigma^2}{K^t m(1-\gamma)}
\end{align}

Taking summation from $t = 0$ to $t = T-1$, we get,
\begin{align}
\label{t1}
    \frac{1}{T}\sum_{t=0}^{T-1}\mathbb{E}\left[\left \| 
    \nabla F(\mathbf{w}^t) \right \|_2^2\right] &\leq \frac{1}{T}\sum_{t=0}^{T-1}\frac{2(\E{F(\mathbf{w}^{t})} - \E{F(\mathbf{w}^{t+1})})}{\eta U (1-\gamma)} + \frac{1}{T}\sum_{t=0}^{T-1}\frac{L\eta\sigma^2}{K^t m(1-\gamma)} \nonumber\\
    &\overset{(a8)}{\leq} \frac{2(\E{F(\mathbf{w^0})}-\E{F(\mathbf{w}^t)})}{T\eta U(1-\gamma)} + \frac{1}{T}\sum_{t=0}^{T-1}\frac{L\eta \sigma^2}{K^0m(1-\gamma)} \nonumber \\
    &\overset{(b8)}{\leq} \frac{2(F(\mathbf{w^0})-F^*)}{T\eta U(1-\gamma)} + \frac{L\eta \sigma^2}{K^0m(1-\gamma)}.
\end{align}

Here ($a8$) follows from $K^0 \leq K^t$ as $K^t$ is increasing, and step ($b8$) follows since we assume $\mathbf{w}^0$ to be known and also from $\E{F(\mathbf{w}^t)} \geq F^*$. Now we complete the proof.
\end{proof}

\subsection{Proof of Corollary \ref{coll1}}
\label{A3}
\begin{proof}
We assume $\eta < \frac{1}{L(\frac{M_G}{m K^t} + \frac{1}{K^t})}$, and let $f(\eta) = \frac{2(F(\mathbf{w^0})-F^*)}{T\eta U(1-\gamma)} + \frac{L\eta \sigma^2}{K^0m(1-\gamma)}$. Then
\begin{align}
    \frac{d f(\eta)}{d \eta} &= -\frac{2(F(\mathbf{w^0})-F^*)}{T\eta^2 U(1-\gamma)} + \frac{L \sigma^2}{K^0m(1-\gamma)}.
\end{align}
Make $\frac{d f(\eta)}{d \eta} = 0$, we get that
\begin{equation}
    \eta' = \sqrt{\frac{2(F(\mathbf{w}^0)-F^*)K^0 m}{T U L\sigma^2}}.
\end{equation}
Note that when $T\rightarrow \infty$, $\eta' \rightarrow 0$ can always satisfy $\eta' < \frac{1}{L(\frac{M_G}{m K^t} + \frac{1}{K^t})}$. Bring $\eta'$ to (\ref{t1}), we get that
\begin{equation}
    \frac{1}{T}\sum_{t=0}^{T-1}\mathbb{E}\left[\left \| \nabla F(\mathbf{w}^t) \right \|_2^2\right] \leq
    \frac{2}{1-\gamma}\sqrt{\frac{2(F(\mathbf{w}^0)-F^*) L \sigma^2}{ T U K^0 m }}.\nonumber
\end{equation}
\end{proof}

\end{document}